\documentclass{llncs}

\title{Predicate Transformers, (co)Monads and Resolutions}

       \author{Pierre Hyvernat\inst{1,2}}
\institute{Institut math\'ematique de Luminy, Marseille, France\\
           \and
           Chalmers Institute of Technology, G\"oteborg, Sweden\\
           \email{hyvernat@iml.univ-mrs.fr}}

\newcommand\ie{\mbox{\textit{i.e.}}~}

\newcommand{\Id}{\mathbf{Id}}
\newcommand{\Fix}{\mathbf{Fix}}
\newcommand{\Pow}{\mathcal{P}}
\newcommand\CONV[1]{{#1}^{\sim}}
\newcommand\AngelU[1]{\langle#1\rangle}
\newcommand\DemonU[1]{[#1]}
\newcommand\BottomU[1]{\lfloor#1\rceil}
\newcommand\AND{\mathbin{\&}}
\newcommand\meets{\mbox{)\kern-4pt(}}

\newcommand\step[2]{\smallbreak\hskip.25cm$#1$\ifx\empty#2\else\hskip.25cm$\{$ {\small #2} $\}$\fi\smallbreak}

\newcommand\be{\[\begin{array}[t]{lllllll}} \newcommand\ee{\end{array}\]}

\begin{document}

    %
%
%

%
%
%

\count255\catcode`@
\catcode`@=11
\chardef\mathlig@atcode\count255

\let\ea=\expandafter

\def\actively#1#2{\begingroup\uccode`\~=`#2\relax\uppercase{\endgroup#1~}}

\def\mathlig@gobble#1{\mathlig@next@cmd}

\def\mathlig@delim{\mathlig@delim}

\def\mathlig@defcs#1{\ea\def\csname#1\endcsname}

\def\mathlig@let@cs#1#2{\ea\let\ea#1\csname#2\endcsname}

\def\mathlig@appendcs#1#2{\ea\edef\csname#1\endcsname{\csname#1\endcsname#2}}

\def\mathlig#1#2{\mathlig@checklig#1\mathlig@end\mathlig@defcs{mathlig@back@#1}{#2}\ignorespaces}


\def\mathlig@checklig#1#2\mathlig@end{%
 \ea\ifx\csname mathlig@forw@#1\endcsname\relax
  \ea\mathchardef\csname mathlig@back@#1\endcsname=\mathcode`#1%
   \mathcode`#1"8000\actively\def#1{\csname mathlig@look@#1\endcsname}%
    \mathlig@dolig#1\mathlig@delim
     \fi
      \mathlig@checksuffix#1#2\mathlig@end}

\def\mathlig@checksuffix#1#2\mathlig@end{%
 \ifx\mathlig@delim#2\mathlig@delim\relax\else\mathlig@checksuffix@{#1}#2\mathlig@end\fi}

\def\mathlig@checksuffix@#1#2#3\mathlig@end{%
 \ea\ifx\csname mathlig@forw@#1#2\endcsname\relax\mathlig@dosuffix{#1}{#2}\fi
  \mathlig@checksuffix{#1#2}#3\mathlig@end}


\def\mathlig@dosuffix#1#2{%
\mathlig@appendcs{mathlig@toks@#1}{#2}%
\mathlig@dolig{#1}{#2}\mathlig@delim
}


\def\mathlig@dolig#1#2\mathlig@delim{%
 \mathlig@defcs{mathlig@look@#1#2}{%
 \mathlig@let@cs\mathlig@next{mathlig@forw@#1#2}\futurelet\mathlig@next@tok\mathlig@next}%
 \mathlig@defcs{mathlig@forw@#1#2}{%
  \mathlig@let@cs\mathlig@next{mathlig@back@#1#2}%
  \mathlig@let@cs\checker{mathlig@chck@#1#2}%
  \mathlig@let@cs\mathligtoks{mathlig@toks@#1#2}%
  \ea\ifx\ea\mathlig@delim\mathligtoks\mathlig@delim\relax\else
  \ea\checker\mathligtoks\mathlig@delim\fi
  \mathlig@next
 }%
 \mathlig@defcs{mathlig@toks@#1#2}{}%
 \mathlig@defcs{mathlig@chck@#1#2}##1##2\mathlig@delim{%
  \ifx\mathlig@next@tok##1%
   \mathlig@let@cs\mathlig@next@cmd{mathlig@look@#1#2##1}\let\mathlig@next\mathlig@gobble
  \fi 
  \ifx\mathlig@delim##2\mathlig@delim\relax\else
   \csname mathlig@chck@#1#2\endcsname##2\mathlig@delim
  \fi
 }%
%
 \ifx\mathlig@delim#2\mathlig@delim\else
  \mathlig@defcs{mathlig@back@#1#2}{\csname mathlig@back@#1\endcsname #2}%
 \fi
}%

\catcode`@\mathlig@atcode

\mathlig{->}{\to}
\mathlig{|->}{\mapsto}
\mathlig{<=>}{\Leftrightarrow}
\mathlig{=>}{\Rightarrow}
\mathlig{<=}{\Leftarrow}
\mathlig{×}{\times}
\mathlig{·}{\cdot}
\mathlig{||}{\ \mid\ }
\mathlig{==}{\equiv}

\maketitle

\begin{abstract}
  This short note contains random thoughts about a factorization theorem for
  closure/interior operators on a powerset which is reminiscent to the notion
  of \emph{resolution} for a monad/comonad. The question originated from
  formal topology but is interesting in itself.

  The result holds constructively (even if it classically has several
  variations);  but usually not predicatively (in the sense that the
  interpolant will no be given by a set). For those not familiar with
  predicativity issues, we look at a ``classical'' version where we bound the
  size of the interpolant.
\end{abstract}


\section*{Introduction} 

A very general theorem states that any monotonic operator $F:\Pow(X) ->
\Pow(Y)$ can be factorized in the form $\Pow(X) -> \Pow(Z) -> \Pow(Y)$, where
$Z$ is an appropriate set; and the first predicate transformer commutes with
arbitrary unions and the second commutes with arbitrary intersections.

We prove similar factorization for interior and closure operators on a
power-set; the idea being to ``resolve'' the operator as is usually done for
(co)monad in categories. We then look at the constructive version of those
factorizations.

\section{Relations and Predicate Transformers} 

We start by introducing the basic notions:
\begin{definition}
  If $X$ and $Y$ are sets, a (binary) relation between $X$ and $Y$ is a subset
  of the cartesian product $X×Y$. The \emph{converse} of a relation
  $r\subseteq X×Y$ is the relation $\CONV{r}\subseteq Y×X$ defined as
  $(y,x)\in\CONV{r} == (x,y)\in r$.

  A \emph{predicate transformer} from $X$ to $Y$ is an operator from the
  powerset $\Pow(X)$ to the powerset $\Pow(Y)$.
\end{definition}
Since most of our predicate transformers will be monotonic (with respect to
inclusion), we drop the adjective when no confusion is possible.

\begin{definition}
  Suppose $r$ is a relation between $X$ and $Y$; we define two monotonic
  predicate transformers from $Y$ to $X$:
  \be
  \AngelU{r} &:& \Pow(Y) &-> & \Pow(X)\\
             & & V       &|->& \{ x\in X || (\exists y\in Y)\ (x,y)\in r \AND y\in V\}
  \ee
  and
  \be
  \DemonU{r} &:& \Pow(Y) &-> & \Pow(X)\\
             & & V       &|->& \{ x\in X || (\forall y\in Y)\ (x,y)\in r => y\in V\} \ \hbox{;}
  \ee
  and an antitonic predicate transformer:
  \be
  \BottomU{r} &:& \Pow(Y) &-> & \Pow(X)\\
             & & V       &|->& \{ x\in X || (\forall y\in V)\ (x,y)\in r\} \ \hbox{.}
  \ee
\end{definition}
Concerning notation:
\begin{itemize}
  \item $\AngelU{r}$ and $\DemonU{r}$ are somewhat common in the refinement
    calculus, even though the main reference (\cite{RC}) uses $\{r\}$ instead
    of $\AngelU{r}$. The problem is that this clashes with set theoretic
    notation.
  \item In \cite{Birkhoff}, Birkhoff uses $V^{\leftarrow}$ for $\BottomU{r}$,
    but this supposes that $r$ is clear from the context. (The notation
    $V^{->}$ would then be $\BottomU{\CONV{r}}(U)$.)
  \item The linear logic community would use $V^\bot$ for the same thing, but
    this also supposes that the relation is called ``$\bot$''.
\end{itemize}

\medbreak
We will always be in a ``typed'' context; \ie subsets will always be subset of
some ambient set. We write $\lnot$ for complementation with respect to that
ambient set.

\begin{lemma} \label{lem:negation}
  Suppose $r$ is a relation between $X$ and $Y$; we have:
  \begin{enumerate}
    \item $\AngelU{r}·\lnot = \lnot·\DemonU{r}$;
    \item $\AngelU{\lnot r} = \lnot·\BottomU{r}$.  \hskip2cm{\footnotesize(where
      $(x,y)\in\lnot r == (x,y)\notin r$)}
  \end{enumerate}
\end{lemma}
A very interesting property is the following Galois connections:
\begin{lemma} \label{lem:Galois}
  Suppose $r$ is a relation between $X$ and $Y$; then
  $\AngelU{r}\vdash\DemonU{r}$, and $\BottomU{r}$ is Galois-connected to
  itself:
  \begin{enumerate}
    \item $\AngelU{r}(V) \subseteq U <=> V \subseteq \DemonU{\CONV{r}}(U)$;
    \item $U\subseteq \BottomU{r}(V) <=> V \subseteq \BottomU{\CONV{r}}(U)$.
  \end{enumerate}
\end{lemma}

Those predicate transformers satisfy:
\begin{lemma} \label{lem:AngelDemonBottom}
  If $r$ is a relation between $X$ and $Y$, then
  \begin{itemize}
    \item $\AngelU{r}$ commutes with arbitrary unions;
      \hfill{\footnotesize(\ie it is a sup-lattice morphism\footnote{All of
      our lattices are complete, so we do not bother writing ``complete'' all
      the time...})}
    \item $\DemonU{r}$ commutes with arbitrary intersections;
      \hfill{\footnotesize(\ie it is an inf-lattice morphism)}
    \item $\BottomU{r}$ transforms arbitrary unions into intersections.
  \end{itemize}
  and moreover:
  \begin{itemize}
    \item any sup-lattice morphism from $\Pow(Y)$ to $\Pow(X)$ is of the form
      $\AngelU{r}$ for some $r\subseteq X×Y$;
    \item any inf-lattice morphism from $\Pow(Y)$ to $\Pow(Y)$ is of the form
      $\DemonU{r}$ for some $r\subseteq X×Y$;
    \item any predicate transformer from $Y$ to $X$ taking arbitrary unions to
      intersections is of the form $\BottomU{r}$ for some $r\subseteq X×Y$.
  \end{itemize}
\end{lemma}
\begin{proof}
The sup-lattice part is easy; and the rest is an application of
Lemma~\ref{lem:negation}.
\qed
\end{proof}

\bigbreak
Just like it is possible to factorize any relation as the composition of a
total function and the inverse of a total function, it is possible to
factorize an monotonic predicate transformer as the composition of a
$\AngelU{r}$ and a $\DemonU{s}$ (see \cite{newalg} for a detailed categorical
construction).
\begin{proposition}  \label{prop:PTfactorize}
  Suppose $F$ is a monotonic predicate transformer from $X$ to $Y$; then there
  is an $X'$ and there are relations $s\subseteq X'×X$ and $r\subseteq Y×X'$
  such that $F = \AngelU{r}·\DemonU{s}$.
\end{proposition}
\begin{proof}
  Let $F$ be a monotonic predicate transformer, and define $X'=\Pow(X)$
  together with $(U,x)\in s == x\in U$ and $(y,U)\in r == y\in F(U)$.
  We have:

  \smallbreak
  $y \in \AngelU{r}·\DemonU{s}(U)$

  \step{<=>}{definition of $\AngelU{r}$}

  $(\exists V \in X')\ (y,V)\in r \AND V\in\DemonU{s}(U)$

  \step{<=>}{definition of $X'$ and $r$}

  $(\exists V\subseteq X)\ y\in F(V) \AND V\in\DemonU{s}(U)$

  \step{<=>}{definition of $\DemonU{s}$}

  $(\exists V\subseteq X)\ y\in F(V) \AND (\forall x)\ (V,x)\in s => x\in U$

  \step{<=>}{definition of $s$}

  $(\exists V\subseteq X)\ y\in F(V) \AND (\forall x)\ V\subseteq U$

  \step{<=>}{since $F$ is monotonic}

  $y\in F(U)$

  \smallbreak
  \noindent
  The result thus holds, but the proof doesn't bring much information...
\qed
\end{proof}

And as a direct application of Lemma~\ref{lem:negation}:
\begin{corollary}
  Any monotonic predicate transformer can be factorized as a
  $\DemonU{r}·\AngelU{s}$ or as a $\BottomU{r}·\BottomU{s}$.
  \\
  Similarly, any antitonic predicate transformer can be written has one of
  $\BottomU{r}·\DemonU{s}$, $\BottomU{r}·\AngelU{s}$, $\AngelU{r}·\BottomU{s}$
  or $\DemonU{r}·\BottomU{s}$.
\end{corollary}

\section{Interior and Closure Operators} 
\label{sec:InteriorClosure}

\begin{definition}  \label{defn:InteriorClosure}
  If $(X,\leq)$ is a partial order, we say that $F:X->X$ is an interior
  operator if:
  \begin{itemize}
    \item $F$ is monotonic;
    \item $F$ is contractive: $F(x)\leq x$;
    \item $F(x) \leq FF(x)$.
  \end{itemize}
  We say that it is a closure operator if:
  \begin{itemize}
    \item $F$ is monotonic;
    \item $F$ is expansive: $x\leq F(x)$;
    \item $FF(x) \leq F(x)$.
  \end{itemize}
\end{definition}

It is well known that the composition of two Galois connected operators yield
interior/closure operators, so that we have:
\begin{lemma} \label{lem:interior}
  If $r$ is a relation between $X$ and $Y$, then
  \begin{itemize}
    \item $\AngelU{r}·\DemonU{\CONV{r}}$ is an interior operator on $\Pow(X)$;
    \item $\DemonU{r}·\AngelU{\CONV{r}}$ is a closure operator on $\Pow(X)$;
    \item $\BottomU{r}·\BottomU{\CONV{r}}$ is a closure operator on $\Pow(X)$.
  \end{itemize}
  Some other consequences of the Galois connection are listed below:
  \begin{itemize}
    \item $\AngelU{r}·\DemonU{\CONV{r}}·\AngelU{r} = \AngelU{r}$;
    \item $\DemonU{r}·\AngelU{\CONV{r}}·\DemonU{r} = \DemonU{r}$;
    \item $\BottomU{r}·\BottomU{\CONV{r}}·\BottomU{r} = \BottomU{r}$.
  \end{itemize}
\end{lemma}

The problem is now to mimic Proposition~\ref{prop:PTfactorize}.
\begin{definition}
  If $F$ is an interior operator on $\Pow(X)$, a \emph{resolution} for $F$ is
  given by a set $Y$ (called the interpolant) together with a relation $r\subseteq Y×X$
  such that $F=\AngelU{r}·\DemonU{\CONV{r}}$.
\end{definition}

\begin{proposition} \label{prop:interior}
  If $F$ is an interior operator on $\Pow(X)$, then it has a resolution.
\end{proposition}
The proof relies on the following lemma:
\begin{lemma} \label{lem:FixInterior}
  Let $F$ be an interior operator on a complete sup-lattice
  $(X,\leq,\bigvee)$; write $\Fix(F)$ for the collection of fixed-point for
  $F$. We have that $(\Fix(F),\leq,\bigvee)$ is a complete sup-lattice; and
  for any $x\in X$
  \[
    F(x) = \bigvee \big\{ y\in\Fix(F) || y\leq x\big\} \ \hbox{.}
  \]
\end{lemma}
\begin{proof} That $(\Fix(F),\leq,\bigvee)$ is a complete sup-lattice is left
  as an easy exercise; for the second point, let $x\in X$;
  \begin{itemize}
    \item we know that $F(x)$ is a fixed point of $F$, and that $F(x)\leq x$.
      This implies that $F(x)\in\{y\in\Fix(F)||y\leq x\}$; and so
      $F(x)\leq\bigvee\{y\in\Fix(F)||y\leq x\}$;
    \item suppose $y\in\Fix(F)$ and $y\leq x$; this implies that $F(y)\leq
      F(x)$, \ie that $y\leq F(x)$. We can conclude that
      $\bigvee\{y\in\Fix(F) || y\leq x\} \leq F(x)$.
  \end{itemize}
\qed
\end{proof}

\begin{proof}[of proposition~\ref{prop:interior}]
  Suppose $F$ is an interior operator on $\Pow(X)$; define $Y=\Fix(F)$ and
  $(U,x)\in r == x\in U$. We have:

  \smallbreak
  $x\in \AngelU{r}·\DemonU{\CONV{r}}(U)$

  \step{<=>}{definition of $r$}

  $\big(\exists V\in\Fix(F)\big)\ x\in V \AND (\forall y)\,y\in V => y\in U$

  \step{<=>}{}

  $\big(\exists V\in\Fix(F)\big)\ x\in V \AND V \subseteq U$

  \step{<=>}{}

  $x\in\bigcup \{ V\in\Fix(F) || V\subseteq U\}$

  \step{<=>}{Lemma~\ref{lem:FixInterior}}

  $x\in F(U)$

  \smallbreak
  \noindent
  which concludes the proof.
\qed
\end{proof}
Just like for Proposition~\ref{prop:PTfactorize}, the statement of the theorem
is interesting, but the proof hardly tells us anything about the structure
of~$F$. To gain a little more information about $F$, we will try to ``bound''
the size of the interpolant set $Y$.

\begin{definition}
  If $(X,\leq,\bigvee)$ is a complete sup-lattice, we say that a family
  $(x_i)_{i\in I}$ of element of $X$ is a basis if, for any $y\in X$, we have
  \[ y = \bigvee \{x_i || x_i\leq y\} \ \hbox{.}\]
\end{definition}

\begin{corollary}
  Suppose $F$ is an interior operator on $\Pow(X)$; if
  $\big(\Fix(F),\subseteq,\bigcup\big)$ has a basis of cardinality $\kappa$,
  then we can find a resolution of $F$ with an interpolant $Y$ of cardinality
  $\kappa$.
\end{corollary}
\begin{proof}
  It is easy to see that in the above proof of
  Proposition~\ref{prop:interior}, we can replace $\Fix(F)$ by a basis of
  $\big(\Fix(F),\subseteq,\bigcup\big)$.
\qed
\end{proof}
In particular, if there is a basis of $\Fix(F)$ which has cardinality less that
the cardinality of $X$; we can use $X$ as the interpolant and use a
relation $r\subseteq X×X$ to obtain a resolution of $F$.

\medbreak
We now show that this result is optimal:
\begin{lemma} \label{lem:BaseFromResolution}
  Let $F$ be an interior operator on $\Pow(X)$; and suppose there are no basis
  of $\Fix(F)$ of cardinality $\kappa$; then there is no interpolant of
  cardinality less than $\kappa$.
\end{lemma}
\begin{proof}
  To show that, we will construct a basis of $\Fix(F)$ indexed by any
  interpolant for $F$. Suppose $Y$ and $r$ form a resolution for $F$.  For any
  $y\in Y$, define $U_y == \AngelU{r}\{y\}$. We will show that $(U_y)_{y\in
  Y}$ is a basis for $\Fix(F)$.

  \smallbreak\noindent
  Each $U_y$ is a fixed point for $F$:

  $\begin{array}{llll}
    U_y &=& \AngelU{r} \{y\}                              & \mskip15mu \{ \hbox{\small definition}\} \\
        &=& \AngelU{r}·\DemonU{\CONV{r}}·\AngelU{r}\{y\}  & \mskip15mu \{ \hbox{\small second part of Lemma~\ref{lem:interior}}\} \\
        &=& F·\AngelU{r}\{y\}                             & \mskip15mu \{ \hbox{\small $r$ is a resolution of $F$}\} \\
        &=& F(U_y)                                        & \mskip15mu \{ \hbox{\small definition}\}
  \end{array}$

  \smallbreak\noindent
  Let $U$ be a fixed point of $F$:

  $\begin{array}{llll}
    U &=& \AngelU{r}·\DemonU{\CONV{r}}(U)                                                &  \mskip15mu \{ \hbox{\small $U$ is a fixed point of $F$} \}\\
    &=& \AngelU{r}\big( \bigcup \big\{ \{y\} || y \in \DemonU{\CONV{r}}(U) \big\}  \big) &  \mskip15mu \{ \hbox{\small basic logic} \}\\
    &=& \bigcup \big\{ \AngelU{r}\{y\} || \{y\} \subseteq \DemonU{\CONV{r}}(U) \big\}    &  \mskip15mu \{ \hbox{\small $\AngelU{r}$ commutes with unions} \} \\
    &=& \bigcup \big\{ \AngelU{r}\{y\} || \AngelU{r}\{y\} \subseteq U\big\}              &  \mskip15mu \{ \hbox{\small Galois connection between $\AngelU{r}$ and $\DemonU{\CONV{r}}$}\} \\
    &=& \bigcup \{ U_y || U_y\subseteq U\}                                               &  \mskip15mu \{ \hbox{\small definition}\}
  \end{array}$

  \noindent
  which concludes the proof that $(U_y)_{y\in Y}$ is a basis for $\Fix(F)$.
\qed
\end{proof}

\bigbreak
Let's look at an example of interior operator on $\Pow(X)$ which \emph{cannot}
be resolved using $X$ as an interpolant. Let $X$ be a countable infinite set
(natural numbers for example); and define $F:\Pow(X)->\Pow(X)$ as follows:
\[
F(U) = \left\{\begin{array}{ll}
                \emptyset & \hbox{if $U$ is finite}\\
                U         & \hbox{if $U$ is infinite}
              \end{array}\right.
\]
The sup-lattice $\Fix(F)$ is given by the collection of infinite subsets of
$X$; and this lattice doesn't have a countable basis. To prove that, it is
enough to do it for any particular countable infinite set. Take $C$ to be the
set of finite strings over $\{0,1\}$. If $\alpha$ is an infinite string of
$0$'s and $1$'s, define $U_\alpha \subseteq C$ to be the set of finite
prefixes of $\alpha$. Each $U_\alpha$ is an element of $\Fix(F)$; but no
countable family of infinite subsets can ``generate'' all the $U_\alpha$:
since $\alpha\neq\beta$ implies that $U_\alpha \cap U_\beta$ is finite, if
$V_i \subseteq U_\alpha$ and $V_j \subseteq U_\beta$ then $i\neq j$. In other
words, a family which generates all the $U_\alpha$'s needs to have the
cardinality of the collection of the $U_\alpha$'s, \ie uncountable.

\bigbreak
Using Lemma~\ref{lem:negation}, we can now extend all what has been done for
interior operators for closure operators: if $F$ is a closure operator on $\Pow(X)$,
then $\lnot·F·\lnot$ is an interior operator on $\Pow(X)$.

\begin{corollary} \label{cor:closure}
  If $F$ is a closure operator on $\Pow(X)$, then it has a resolution as a
  composition $\DemonU{r}·\AngelU{\CONV{r}}$ or as
  $\BottomU{r}·\BottomU{\CONV{r}}$.

  As for interior operators, the possible cardinalities of the interpolant are
  given by the cardinalities of the bases for the inf-lattice $\Fix(F)$.
\end{corollary}

In particular, for linear logicians, it is not the case that any closure
operator can be written as a biorthogonal...

\section{Comonad and Monads} 

It seems that the traditional way to look at monads in a category is to see
them as a kind of generalized monoid; at least in my part of the world.
Another view\footnote{which has given me a much better understanding of what
(co)monads are} is to view them as a generalization of closure operators. This
is the view taken in the introduction of~\cite{LambekScott}.

\begin{definition}
  A monad on a category $\mathcal{C}$ is a morphism $F:\mathcal{C} ->
  \mathcal{C}$ together with two natural transformations
  $\eta:\Id_{\mathcal{C}} -> F$ and $\mu:FF->F$ s.t. some diagram commute.
\end{definition}
If one takes the partial order category $\Pow(X)$, then a monad is an operator
on $\Pow(X)$ s.t.:
\begin{itemize}
  \item it acts on morphisms: if $i:U\subseteq V$ then $F_i:F(U)\subseteq
    F(V)$; \ie $F$ is monotonic;
  \item there is a natural transformation $\eta_U:U\subseteq F(U)$;
  \item there is a natural transformation $\mu_U:FF(U)\subseteq F(U)$.
\end{itemize}
All the ``coherence'' conditions are trivially satisfied in a partial order
category, since every diagram commute!
What we've just shown is that a monad on $\Pow(X)$ is nothing more than a
closure operator on $\Pow(X)$; and vice and versa.
Similarly, a comonad corresponds to an interior operator.

\medbreak
In categories, a \emph{resolution} for a monad corresponds to factorizing the
functor $F$ as the composition of two adjoint functors. Adjointness $H\vdash G$
between the functors $G:\mathcal{D}->\mathcal{C}$ and
$H:\mathcal{C}->\mathcal{D}$ in a locally small category means that~$\mathcal{C}[A,G(B)] \simeq \mathcal{D}[H(A),B]$ which, in the case of a
partial order category simplifies to ``$U\subseteq G(V)$ iff $H(U)\subseteq
V$'' \ie is exactly the Galois connection~$H\vdash G$.

Category theory tells us that there always is a resolution since we have two
degenerate resolutions: something like $F=\Id·F$ and $F=F·\Id$. The first
one is given by the Eleinberg-Moore category. In the context of order theory,
it amounts to taking: interpolant category to partial orders: if $F$ is a
closure operator on~$X$, take $\mathcal{E}(X,F)$ to be the the collection of
fixed points for $F$, with the ordering inherited from $X$.\footnote{In a
partial order, an algebra for $F$ is just a post fixed point: $x\leq F(x)$. If
$F$ is a closure, then it is also a fixed point for $F$.} We have~$F:X->\mathcal{E}(X,F)$ and~$\Id:\mathcal{E}(X,F)->X$
which are adjoint: $F\vdash\Id$.

 $x \leq \Id(f)$

 \step{=>}{$F$ is monotonic}

 $F(x)\leq F(f)$

 \step{<=>}{$f\in\mathcal{E}(X,F)$, \ie $f$ is a fixed point for $F$}

 $F(x) \leq f$

 \noindent
 and $F(x)\leq f => x\leq f=\Id(f)$ since $x\leq F(x)$.

\smallbreak
The second resolution ($F·\Id$) is obtained via the Kleisli categry. In the
context of order theory, it amounts to taking the order~$(X,\sqsubseteq)$
where~$x\sqsubseteq y$ iff~$F(x)\leq F(y)$.

What is important to us is that those resolutions are trivial and
uninteresting. More abstract nonsense states that there is a category of
resolutions for any given monad; and that those two trivial resolutions are
respectively initial/terminal. What we have done with
Proposition~\ref{prop:interior} and Corollary~\ref{cor:closure} corresponds to
finding a non-trivial \emph{resolution} for the monad corresponding to the
interior/closure operator.
It is even more than a resolution, since the interpolant is itself a complete
and cocomplete category (and one of the functors preserves limits while the
other one preserves colimits; but this is a general fact about adjoints).  The
feeling is that this resolution lies ``exactly in the middle'' between the
initial and terminal resolutions. I don't know if this kind of ``strong''
resolution has been considered in category
theory.\footnote{\label{foot:strongResolution}\ie if $\mathcal{C}$ is a
complete and cocomplete category and $F$ is a monad on $\mathcal{C}$, a
``strong resolution'' is a resolution with a complete and cocomplete
interpolant. Any reference to something similar in the literature would be
most welcome.}

\medbreak
The resolution constructed here is quite different from the Eleinberg-Moore
resolution (even though it uses the fixed-points as a basis). We do not
construct functors from $\Pow(X)$ to $\mathcal{E}(X,F)$ and back; but from
$\Pow(X)$ to $\Pow\big(\mathcal{E}(X,F)\big)$ and back. In particular, as
noted above, the interpolant is complete and cocomplete;\footnote{\ie is a
complete lattice since we deal with partial orders} which is not the case for
the Eleinberg-Moore category: fixed points for an interior are closed under
unions but not under intersections; and conversely for closure operators.

\section{Revisiting Section~\ref{sec:InteriorClosure} in a Constructive Setting} 

\subsection{Impredicative}

In the previous section we used Lemma~\ref{lem:negation} to
generalize results on interior to closures. It is thus natural to ask whether
(1) we can make the original proof constructive; (2) we can avoid using this
lemma and prove the result for closure constructively. I will not into the
details but just provide some hints about that.

\begin{itemize}
  \item Galois connections from Lemma~\ref{lem:Galois} are constructive.
  \item Lemma~\ref{lem:AngelDemonBottom}: the first part is trivial. The
    second part is easy: if $F$ commutes with unions, take $(x,y)\in r$ iff
    $y\in F\{x\}$; if $F$ commutes with intersections, take $(x,y)\in r$ iff
    $(\forall U)\, y\in F(U) => x\in U$; and if $F$ transforms unions into
    intersections, take $(x,y)\in r$ iff $y\in F\{x\}$.
   \item Proposition~\ref{prop:PTfactorize}: the proof is constructive.
   \item Lemma~\ref{lem:AngelDemonBottom} is constructive.
   \item Lemma~\ref{lem:FixInterior} and the corresponding lemma for closure
     operator are constructive.
   \item Proposition~\ref{prop:interior} is constructive.
   \item we can mimic the proof of Proposition~\ref{prop:interior} to obtain a
     resolution of a closure operator as $F=\BottomU{r}·\BottomU{\CONV{r}}$,
     but \emph{not} to obtain a resolution as
     $F=\DemonU{r}·\AngelU{\CONV{r}}$.
   \item I doubt we can constructively obtain a resolution of a closure
     operator as $F=\DemonU{r}·\AngelU{\CONV{r}}$.
   \item all of the lemmas about the size of interpolant are constructive at
     least if we read then as ``if $B$ is a basis for $\Fix(F)$ then we can use
     $B$ as an interpolant''.
\end{itemize}

As a proof of concept, all this (except the last point) has been proved in the
proof assistant COQ.\footnote{proof scripts available from
\texttt{http://iml.univ-mrs.fr/\~{}hyvernat/academics.html}}

\subsection{Predicative}

In a predicative setting like Martin-L\"of type theory (see \cite{ML84}) or
CZF (constructive ZF, see \cite{CZF}) set theory, many of the results are not
provable. The reason being that we do not allow quantification on a
power-set. The main result about resolution would become something like:
\begin{proposition}
  Suppose $\Fix(F)$ (proper type) has a set-indexed basis, then $F$ as a
  resolution as $\AngelU{r}·\DemonU{\CONV{r}}$ (if $F$ is an interior) or as
  $\BottomU{r}·\BottomU{\CONV{r}}$ (if $F$ is a closure).
\\
  If $F$ as a resolution, then $\Fix(F)$ as a set-index basis.
\end{proposition}
Note that having a set-indexed basis is equivalent to being ``set presented''
in the terminology of P. Aczel (\cite{CZF,Aczel}). Note that in the case of an
interior operator, this implies that the predicate transformer is
``set-based'' (in the sense that it has a factorization as in
Proposition~\ref{prop:PTfactorize}, where the interpolant $X'$ is a set).
It doesn't seem that the existence of a resolution for a closure operator
implies that the original predicate transformer is itself set-presented.


\section*{Conclusion} 

Nothing revolutionary has really been done, but the statement of
Propositions~\ref{prop:PTfactorize} and~\ref{prop:interior} is, in an abstract
setting, quite neat.
However, as the proofs show, this is
mostly abstract nonsense. The best example is probably
Proposition~\ref{prop:PTfactorize}, where $y\in F(U)$ is factorized as ``there
is a $V$ such that $s\in F(V)$ and $V\subseteq U$''. The proof of
Proposition~\ref{prop:interior} is slightly subtler, but is hardly
interesting. In the end, the most interesting and informative thing is
probably Lemma~\ref{lem:BaseFromResolution}, in its ``positive'' version: if
$F$ has $Y$ as an interpolant, then $\Fix(F)$ has a basis indexed by $Y$, which
is hardly a breakthrough in mathematics...

I do nevertheless hope that it might interest some people, since while
Proposition~\ref{prop:PTfactorize} is known to many (especially the refinement
calculus people), it seems that Proposition~\ref{prop:interior} isn't stated
anywhere. I also hope the link between monad and closure operator (together
with the Eleinberg-Moore category being the partial order of fixed points)
will gain in popularity, as I see it as a much better way of seeing monads, at
least as far as intuition is concerned.

\medbreak
A final word about the motivation for this: the starting point was the
question about whether it is possible to represent any ``basic topology'' (see
the forthcoming \cite{BP3}) as the formal side of a basic pair,
impredicatively speaking.  A basic topology is a structure
$(X,\mathcal{A},\mathcal{J})$ where $X$ is a set, and $\mathcal{A}$ and
$\mathcal{J}$ are closure and interior operators on $\Pow(X)$ such that
$\mathcal{A}(U) \meets \mathcal{J}(V) => U
\meets(\mathcal{J}(V)$.\footnote{$U)\meets V$ is the constructive version of
$U\cap V\neq \emptyset$.} The answer is obviously \emph{no} since the
$\mathcal{A}$ and $\mathcal{J}$ arising from a formal pair are classically
dual (\ie $\mathcal{A}·\lnot = \lnot·\mathcal{J}$) and there are basic
topologies which are provably not dual. The question then turned into: ``can
any interior operator be written as the formal interior of a basic pair?'' and
similar for closure operators. The answers are \emph{yes}
(Proposition~\ref{prop:interior}) and \emph{I don't know} (the constructive
resolution of a closure is of the form $\BottomU{r}·\BottomU{\CONV{r}}$, and
not of the form $\DemonU{r}·\AngelU{\CONV{r}}$.)


\bigbreak
\bibliographystyle{splncs}
\bibliography{resolutions}

\end{document}